\documentclass[preprint,12pt]{elsarticle1}

\usepackage{subfig}
\usepackage{graphicx}
\usepackage{caption,color}   
\usepackage{amssymb}
\usepackage{amsthm}
\usepackage{amsmath}
\usepackage{epic}
\usepackage{setspace}
\usepackage{float}
\usepackage{multirow}

\usepackage{hyperref}
\usepackage{xcolor}
\usepackage{marginnote}
\usepackage{booktabs}
\usepackage{algorithm, algorithmicx, algpseudocode}

\newtheorem{thm}{Theorem}[section]

\newtheorem{defi}[thm]{Definition}

\numberwithin{equation}{section}
\journal{}

\begin{document}
\begin{spacing}{1.15}
\begin{frontmatter}

\title{An $\alpha$-triangle eigenvector centrality of graphs}
\author[label1]{Qingying Zhang }
\author[label1]{Lizhu Sun\corref{cor}} \ead{lizhusun@hrbeu.edu.cn}
\author[label1]{Changjiang Bu}
\cortext[cor]{Corresponding author.}

\address{
\address[label1]{School of Mathematical Sciences, Harbin Engineering University, Harbin 150001, PR China}
}

%

\begin{abstract}
Centrality represents a fundamental research field in complex network analysis, where centrality measures identify important vertices within networks. Over the years, researchers have developed diverse centrality measures from varied perspectives.
This paper proposes an $\alpha$-triangle eigenvector centrality ($\alpha$TEC), which is a global centrality measure based on both edge and triangle structures.
It can dynamically adjust the influence of edges and triangles through a parameter $\alpha$ ($\alpha \in (0,1]$). The centrality scores for vertices are defined as the eigenvector corresponding to the spectral radius of a nonnegative tensor. By the Perron-Frobenius theorem, $\alpha$TEC guarantees unique positive centrality scores for all vertices in connected graphs.
Numerical experiments on synthetic and real world networks demonstrate that $\alpha$TEC effectively identifies the vertex's structural positioning within graphs. As $\alpha$ increases (decreases), the centrality rankings reflect a stronger (weaker) contribution from edge structure and a weaker (stronger) contribution from triangle structure. Furthermore, we experimentally prove that vertices with higher $\alpha$TEC rankings have a greater impact on network connectivity.
\end{abstract}
\end{frontmatter}

\section{Introduction}

Centrality measures aim to identify the most important vertices in networks.
They play significant roles across multiple domains, including social network analysis \cite{landherr2010critical,das2018study}, brain network \cite{joyce2010new,lohmann2010eigenvector}, and network clustering \cite{berahmand2018new}.
Some centrality measures such as degree centrality \cite{albert2000error},
and triangle centrality \cite{burkhardt2024triangle} are local centrality measures, they only consider the local structure.
Eigenvector centrality \cite{bonacich1972factoring}, subgraph centrality \cite{estrada2005subgraph,zhou2023estrada}, two-steps eigenvector centrality \cite{xu2023two}, communicability centrality \cite{de2022communication,el2023perron}, betweenness centrality \cite{barthelemy2004betweenness}, and All-Subgraphs centrality \cite{bugedo2024family} are global centrality measures, they consider the entire network structure.

Different centrality measures focus on distinct types of topological information within networks, and lead to significant discrepancies in the critical vertices identified.
Triangles in complex networks embody close interactions among vertices and hold significant applications in domains such as link prediction \cite{chavan2020higher}, community detection \cite{gao2022graph}, network stability analysis \cite{hu2023triangular}.
There are many articles on network centrality and clustering centered by triangles. \cite{durak2012degree,ma2019local,zhang2020community,arrigo2020framework,burkhardt2024triangle}.
This paper proposes a method to calculate centrality scores for vertices based on triangles.
However, if we only consider triangles may lead to substantial information loss, particularly in failing to differentiate vertices that do not participate in any triangles.

To address this limitation, we propose an $\alpha$-triangle eigenvector centrality ($\alpha$TEC), which is a global centrality measure. This measure integrates both edges and triangles influences on a vertex's importance, with parameter $\alpha$ ($\alpha \in (0,1]$) dynamically adjusting the relative weights of edges and triangles.
First, we construct a tensor based on the edges and triangles in which each vertex participates. Subsequently, we employ the positive eigenvector corresponding to the tensor's spectral radius as the centrality score for vertices.
Specifically, when $\alpha = 1$, centrality is entirely influenced by the edges, and in connected graphs each vertex is assigned a nonzero centrality score. Conversely, when $\alpha = 0$, this centrality totally governed by triangles, and vertices not participating in any triangles are assigned zero scores.
Therefore, in order to assign nonzero centrality scores to all vertices, we restrict $\alpha \in (0,1]$. And we investigate how vertices centrality changes as $\alpha$ varies in subsequent experiment.

Moreover, we conducte experimental analyses in real networks and compare the $\alpha$TEC with degree centrality, triangle centrality, betweenness centrality and subgraph centrality. Through these analyses, we find that $\alpha$TEC can reflect the position of vertices within the graph. When the $\alpha$ value is large, vertices with high rankings are generally located in areas with dense edge structures, and the vertices in their edge neighborhoods also tend to have relatively higher scores. Conversely, when the $\alpha$ value is small, vertices with high rankings are typically found in areas with dense triangular structures, and the vertices in their triangle neighborhoods are also in regions with concentrated triangle distributions. At this point, vertices that are not in triangles but are adjacent to vertices in triangles will have higher scores than those that are neither in triangles nor adjacent to vertices in triangles.

\section{An $\alpha$-triangle eigenvector centrality}
The theoretical part of this paper is based on tensors, the following introduces some fundamental knowledge about tensors.
Let $\mathbb{C}^n$ and $\mathbb{C}^{[k,n]}$ denote the sets of $n$-dimensional vectors and $k$-th order $n$-dimensional tensors over the complex number field $\mathbb{C}$, respectively.
Let $[n]=\{1,2,\cdots,n\}$.
A tensor $\mathcal{A}=(a_{{i_1} {i_2} {\cdots}  {i_k}}) \in \mathbb{C}^{[k,n]} $ is a multi-dimensional array containing $n^k$ elements, where $i_j \in [n], j \in [k]$.
A tensor is nonnegative if  all its elements are nonnegative.
For $\boldsymbol {x}=(x_1,x_2,{\cdots},x_n)^{\mathrm{T}} \in \mathbb{C}^n$,
$\mathcal{A}{\boldsymbol {x}}^{k-1}$ is a vector in $\mathbb{C}^n$,
with the $i$-th component
\begin{align*}
(\mathcal{A}{\boldsymbol {x}}^{k-1})_i=\sum_{{i_2}, {\cdots},  {i_k}=1}^{n}  a_{{i} {i_2} {\cdots} {i_k}} {x_{i_2}}{x_{i_3}}{\cdots}{x_{i_k}}, i \in [n].
\end{align*}
If there exists  $\lambda \in \mathbb{C}$ and a nonzero vector $\boldsymbol {x}$ such that
\begin{align} \label{tzfc}
\mathcal{A}{\boldsymbol {x}}^{k-1}=\lambda \boldsymbol {x}^{[k-1]},
\end{align}
then $\lambda$  is called an eigenvalue of $\mathcal{A}$,
and $\boldsymbol {x}$ is an eigenvector of  $\mathcal{A}$ corresponding to $\lambda$ \cite{qi2005eigenvalues,lim2005singular},
where $\boldsymbol {x}^{[k-1]}=(x_1^{k-1}, {\cdots}, x_n^{k-1})^{\mathrm{T}}$.
The maximum modulus among all eigenvalues of $\mathcal{A}$ is called the spectral radius of $\mathcal{A}$, denoted by $\rho(\mathcal{A})$.

Extensive studies has been conducted by researchers on the spectral problems of tensors \cite{chang2008perron,qi2017tensor,cooper2012spectra,liu2023generalization,chen2024spectra}.
The eigenvector corresponding to $\rho(\mathcal{A})$ has significant applications in various fields \cite{benson2019three,liu2023high,xu2023two,tudisco2021node}.
In this work, we use it as centrality scores for vertices.
Before formally defining centrality, we first introduce several prerequisite definitions.

Let $V_\triangle$ be a set of vertex sets, where each vertex set contains three mutually connected vertices that form a triangle.
For a connected graph $G=(V(G),E(G))$ of $n$ vertices, where $V(G)$ and $E(G)$ represent the set of vertices and the set of edges, respectively,
we define an edge tensor $\mathcal{A}_E=(b_{ijk})$ and a triangle tensor $\mathcal{A}_\triangle=(c_{ijk})$ based on the edge neighborhoods and triangle neighborhoods of vertices in $V(G)$, where $i,j,k \in [n]$,
\begin{align*}
b_{i j k }=
\begin{cases}
1,                                &  if \text{ $ \{i,j\}\in E(G), k = j  $},\\
0 ,                                &  otherwise,
\end{cases}
\end{align*}
and
\begin{align*}
c_{i j k }=
\begin{cases}
\frac{1}{2}  ,&  if \text{ $ \{i, j, k  \} \in V_\triangle  $,}\\
0 ,                                &  otherwise.
\end{cases}
\end{align*}

\begin{defi} \label{tztzl}
For a connected graph $G$,
tensor $\mathcal{A}= \alpha \mathcal{A}_E + (1-\alpha) \mathcal{A}_\triangle$ is called the $\alpha$-triangle tensor of $G$,
where $0 < \alpha \leq 1$.
\end{defi}


For a nonnegative tensor $\mathcal{A}=(a_{{i_1} {i_2} {\cdots}  {i_k}})$, ${i_j} \in [n]$, $j \in [k]$,
let $D_{\mathcal{A}}=(V{(D_{\mathcal{A}})},E{(D_{\mathcal{A}})})$ be the associated directed graph of $\mathcal{A}$,
with the vertex set $V{(D_{\mathcal{A}})}=[n]$, and the arc set
$E{(D_{\mathcal{A}})}=\left\{ {(i,j)| a_{i {i_2}{\cdots}{i_k} } \neq 0 ,  j \in {\left\{{{i_2},{\cdots},{i_k}}\right\}}  }\right\}$.
For any distinct $i,j \in {V(D_{\mathcal A})}$,
if there exists a directed path from  $i$ to $j$ and $j$ to $i$,
then $D_{\mathcal{A}}$  is said to be strongly connected.
The nonnegative tensor $\mathcal{A}$ is weakly irreducible if and only if $D_{\mathcal{A}}$ is strongly connected \cite{friedland2013perron}.

\begin{thm} \label{irreducible}
For a connected graph $G$, the $\alpha$-triangle tensor $\mathcal{A}$ of $G$ is weakly irreducible
($\alpha \in (0,1]$).
\end{thm}

\begin{proof}
Let $G=(V(G),E(G))$ be a graph with $n$ vertices.
And let $D_{\mathcal{A}}$ be the associated directed graph of $\mathcal{A}$.
For $i_1,i_2 \in [n]$,
if $\{i_1,i_2\}\in E(G)$,
we have $a_{i_1i_2i_2}\neq 0$ in $\mathcal{A}$.
Thus, in $D_{\mathcal{A}}$ there exists a directed arc from $i_1$ to $i_2$.
Because $G$ is connected,
then for any $i,j\in [n]$ there is a directed path from $i$ to $j$ in $D_{\mathcal{A}}$.
So $D_{\mathcal{A}}$ is strongly connected,
we can know that $\mathcal{A}$ is weakly irreducible.
\end{proof}

From the Perron-Frobenius theorem for nonnegative weakly irreducible tensors \cite{friedland2013perron},
we have $\boldsymbol {x}$ is the unique positive eigenvector corresponding to $\rho (\mathcal{A})$ (up to its scalar multiples).
Since Theorem \ref{irreducible}, the $\alpha$-triangle tensor of connected graph $G$ is weakly irreducible,
and the positive eigenvector associated with $\rho(\mathcal{A})$ is uniquely determined.

\begin{defi}
For a connected graph $G$,
let $\mathcal{A}$ be the $\alpha$-triangle tensor of $G$, and $\rho(\mathcal{A})$ is the radius of $\mathcal{A}$.
Then the eigenvector $\boldsymbol {x}$ 
 corresponding to $\rho(\mathcal{A})$ is called the
$\alpha$-triangle eigenvector centrality of $G$.
\end{defi}

Then for a vertex $i \in V(G)$, $x_i > 0 $ is the $\alpha$-triangle eigenvector centrality ($\alpha$TEC) score of $i$.
And from Definition \ref{tztzl} and Equ. (\ref{tzfc}), we have
\begin{align} \label{88}
\rho(\mathcal{A}) x_i^{2} =
\alpha \sum_{\{i,j\}\in E(G)}  x_j^{2}+
(1-\alpha) \sum_{\{ i,j,k \} \in V_\triangle }  x_{j} x_{k}.
\end{align}
As evident from Equ. (\ref{88}), the $\alpha$TEC of vertex $i$ is jointly influenced by vertices in both its edge neighborhood and triangle neighborhood,
with the relative weights of these influences determined by $\alpha$.
Under some centrality measures, the centrality score of a vertex is $0$,
while $\alpha$TEC assigns nonzero centrality scores to all vertices.

\section{Some numerical examples}
Here we introduce several classical centrality measures. In the following experiments, we compare the rankings of vertices under these centralities with the rankings under the $\alpha$-triangle eigenvector centrality ($\alpha$TEC).
Degree centrality (DC) quantifies a vertex's significance through the number of edges that contain it \cite{albert2000error}.
The triangle centrality (TC) \cite{burkhardt2024triangle} of a vertex is assigned by calculating the number of triangles in which the vertex and its adjacent vertices are in.
Betweenness centrality (BC) measures a vertex's importance by counting the shortest paths passing through it \cite{freeman1991centrality}.
Subgraph centrality (SC) evaluates vertex influence using closed walks of fixed length $k$ originating from the vertex \cite{estrada2005subgraph}.
Eigenvector centrality (EC) assigns centrality scores to vertex based on other vertices that share an edge with it \cite{bonacich1972factoring}.
Note that, under $\alpha$TEC, the vertices ranking is same as EC when $\alpha =1$.
\subsection{An simple example}
To further demonstrate that the $\alpha$TEC ranking dynamics of vertices as the $\alpha$ value varies reflect their structural positioning within the graph, we conduct a detailed analysis in an simple example.
Figure \ref{fig222} depicts a graph and its $\alpha$TEC scores under different $\alpha$ values.
Additionally, we list the centrality scores for some centrality measures and $\alpha$TEC under different $\alpha$ values in Table \ref{table111}.

As can be seen from Table \ref{table111}, under other centrality measures, either vertex $4$ or $8$ consistently ranks first, and some vertices have centrality scores of $0$.
Moreover, the vertex rankings under DC, SC, and BC fail to capture the influence of triangles in the graph. Under TC, vertex $4$ (adjacent to two triangles) ranks first, vertices within a triangle are all ranked second, and the remaining vertices score zero. This approach loses some important information and fails to distinguish the detailed positions of vertices in the graph.

\begin{figure}[H]
\centering
\subfloat[]{
\includegraphics[width=0.37\textwidth]{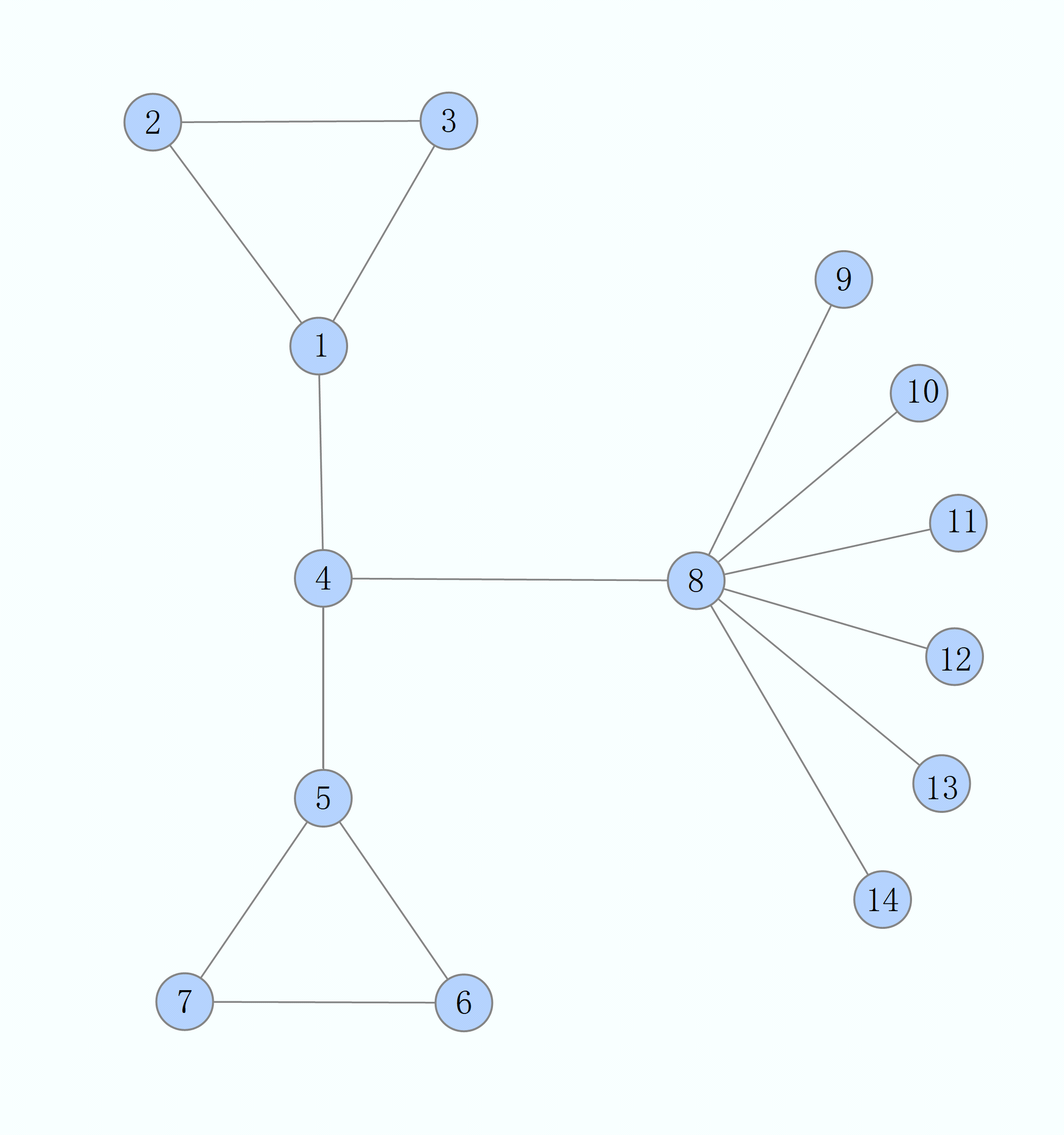}
}
\subfloat[]{
\includegraphics[width=0.48\textwidth]{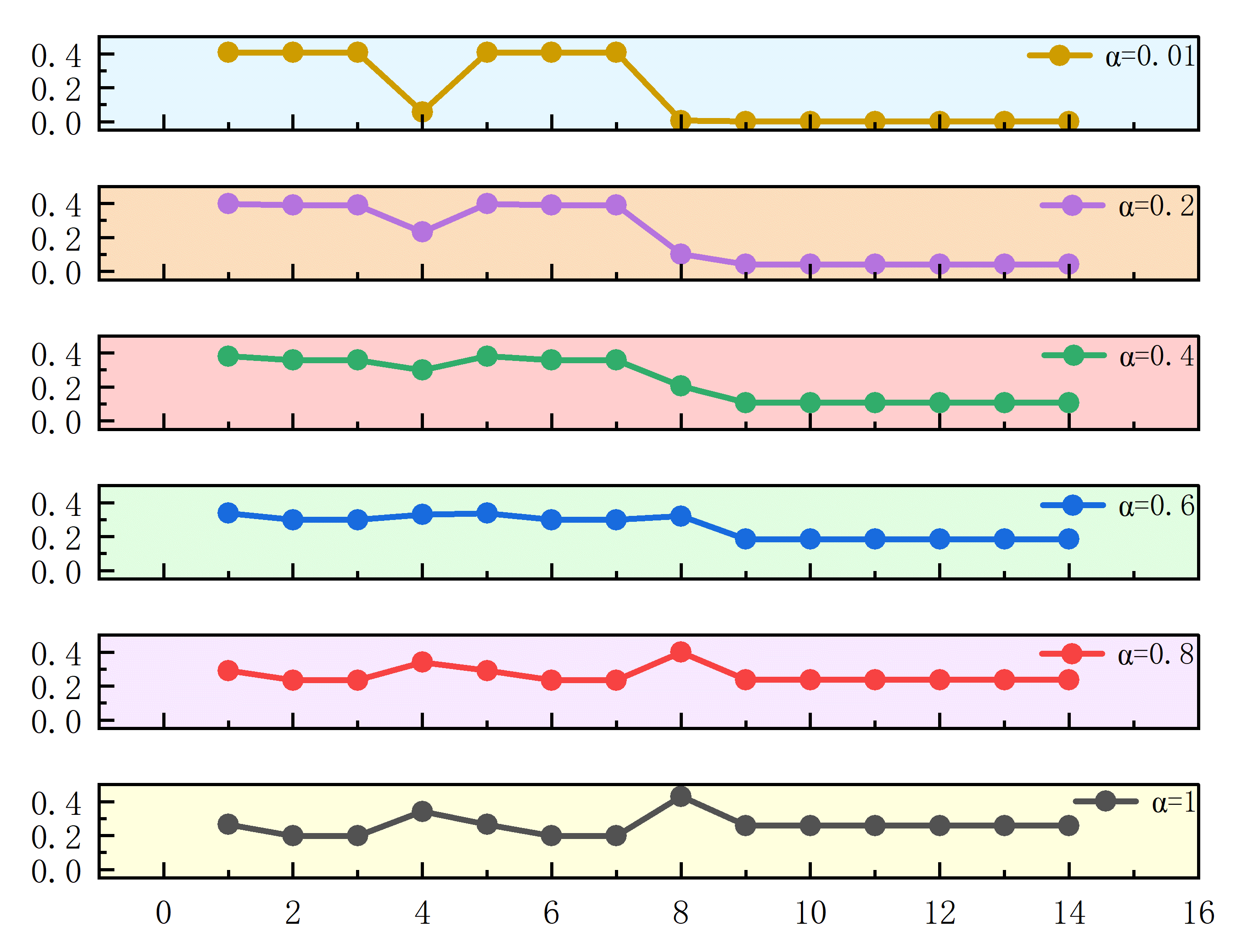}
}
\caption{(a) is a graph with $14$ vertices and (b) is vertices's $\alpha$TEC scores distribution under different values of $\alpha$.}
\label{fig222}
\end{figure}

\begin{table*}
    \tiny
    \centering
    \caption{The centrality scores under some centrality measures and under different $\alpha$ values of $\alpha$TEC.}
    \label{table111}
    \begin{tabular}{lccccc}  
    \toprule
        \textbf{Measures} & \textbf{1,5} & \textbf{2,3,6,7} & \textbf{4} & \textbf{8} & \textbf{9,10,11,12,13,14} \\ \midrule
        DC            & 3 & 2 & 3 & \textbf{7} & 1 \\
        TC           & 0.4082 & 0.4082 & \textbf{0.8165} & 0& 0 \\
        SC            
        & 0.3023 & 0.2330 & 0.2907  & \textbf{0.6042} & 0.1565\\
        BC           
        & 0.2664 & 0 & 0.6176 & \textbf{0.6903} & 0 \\
        $\alpha = 1$      & 0.2651 & 0.1984 & 0.3421 & \textbf{0.4306} & 0.2580 \\
        $\alpha = 0.8$      & 0.2905 & 0.2348 & 0.3398 & \textbf{0.3988} & 0.2367 \\
        $\alpha = 0.6$      & \textbf{0.3379} & 0.2984 & 0.3303 & 0.3209 & 0.1842 \\
        $\alpha = 0.4$      & \textbf{0.3801} & 0.3578 & 0.2984 & 0.2054 & 0.1064 \\
        $\alpha = 0.2$      & \textbf{0.3977} & 0.3904 & 0.2312 & 0.1019 & 0.0411 \\
        $\alpha = 0.01$     & \textbf{0.4076} & \textbf{0.4076} & 0.0573 & 0.0058 & 0.0005 \\
    \bottomrule
    \end{tabular}
\end{table*}


As shown in Figure \ref{fig222} (b) and Table \ref{table111},
vertices $1$ through $6$ are in a triangle and other vertices not in any triangle,
the rankings of vertices exhibit significant changes under $\alpha$TEC as the value of $\alpha$ decreases from $1$ to $0.01$.
When $\alpha$ equals $1$ or $0.8$, edges exert strong influence on centrality, vertex $8$ rank first while vertices $2, 3, 6$, and $7$ consistently rank last.
As $\alpha$ gradually decreases from $0.6$ to $0.01$, the influence of triangles progressively strengthens.
Throughout this process, vertices $1$ and $5$ maintain the highest rank, whereas scores for vertices $4$ and $8$ decline slowly. In contrast, vertices $9$ through $14$ exhibit a sharp score decrease.
This occurs because vertex $4$ is influenced by edges from vertices $1$ and $5$, while vertex $8$ is affected by edges from vertex $4$.
Consequently, except when $\alpha=0.01$, vertices $1$ and $5$ consistently rank higher than vertices $2, 3, 6$, and $7$.
\subsection{Some real networks}
In the following, we compute the $\alpha$TEC of vertices in Zachary's karate club network \cite{zachary1977information}, Lusseau's Dolphin network \cite{lusseau2003bottlenose} and C.elegans metabolic network \cite{rossi2015network} for different values of $\alpha$.
The Zachary's karate club is a network with $34$ vertices and $68$ edges, primarily divided into two communities with little difference in the internal density of the two communities. The Lusseau's Dolphin is a network with $62$ vertices and $159$ edges, mainly divided into two communities, but the internal density of one community is significantly higher than the other. The C.elegans metabolic is a network with $453$ vertices and $2025$ edges, where some vertices are very closely connected, much more so than other vertices.

We plot graphs of the three networks and the distribution of centrality scores under different $\alpha$ values, respectively, as shown in Figure \ref{fig3}.
It can be observed that the centrality scores distributions demonstrate significant divergence across three networks.
And the rankings of certain vertices in three networks exhibit significant variations under different $\alpha$ values.
We subsequently conduct an in-depth analysis of centrality variations across the three networks from three perspectives.

\begin{figure}[H]
  \centering
   \subfloat[Zachary's karate club network.]{%
    \begin{minipage}[t]{0.47\textwidth}
      \centering
      \includegraphics[width=\linewidth]{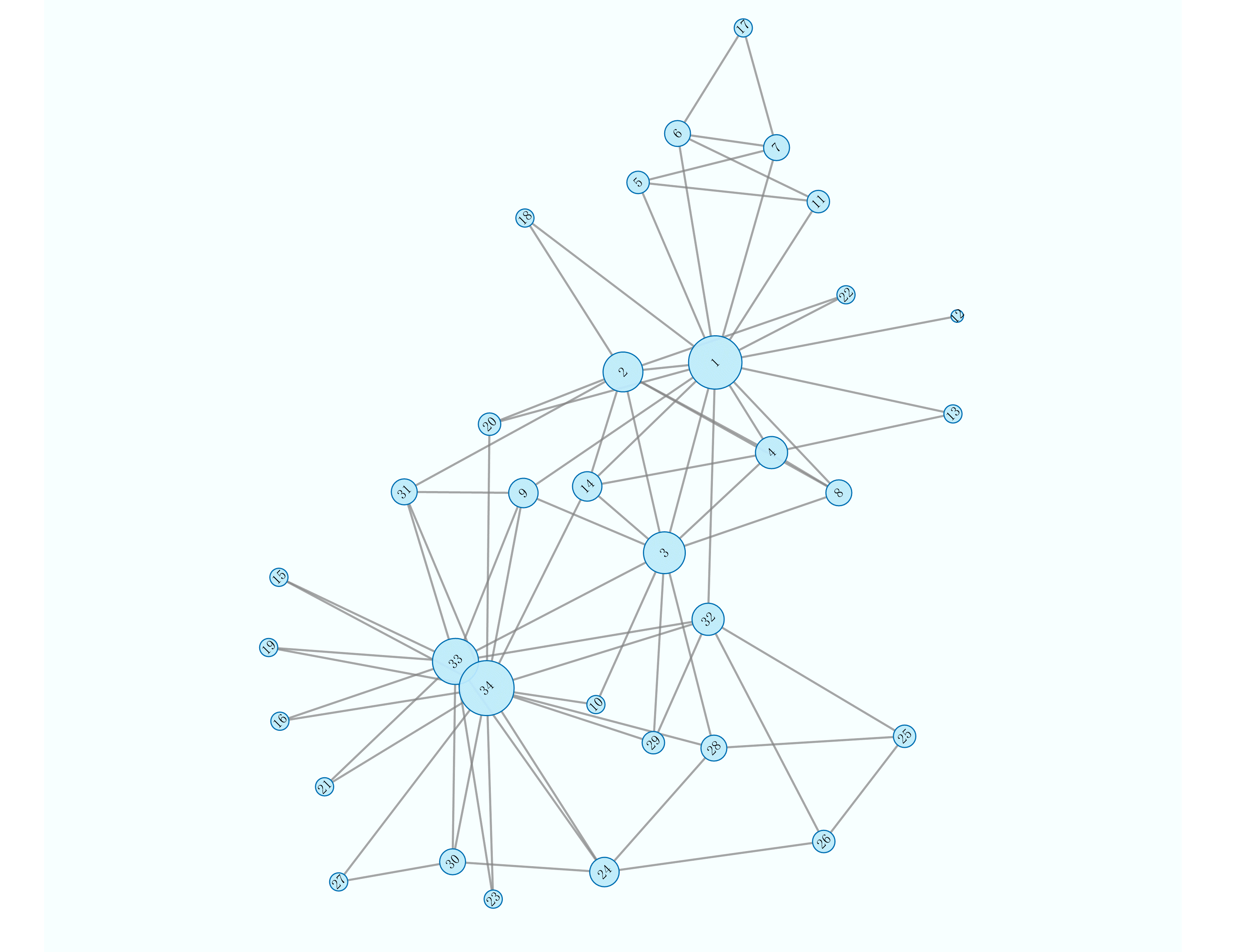}
    \end{minipage}%
  }
  \subfloat[Centrality scores of vertices under different $\alpha$ values.]{%
    \begin{minipage}[t]{0.47\textwidth}
      \centering
      \includegraphics[width=\linewidth]{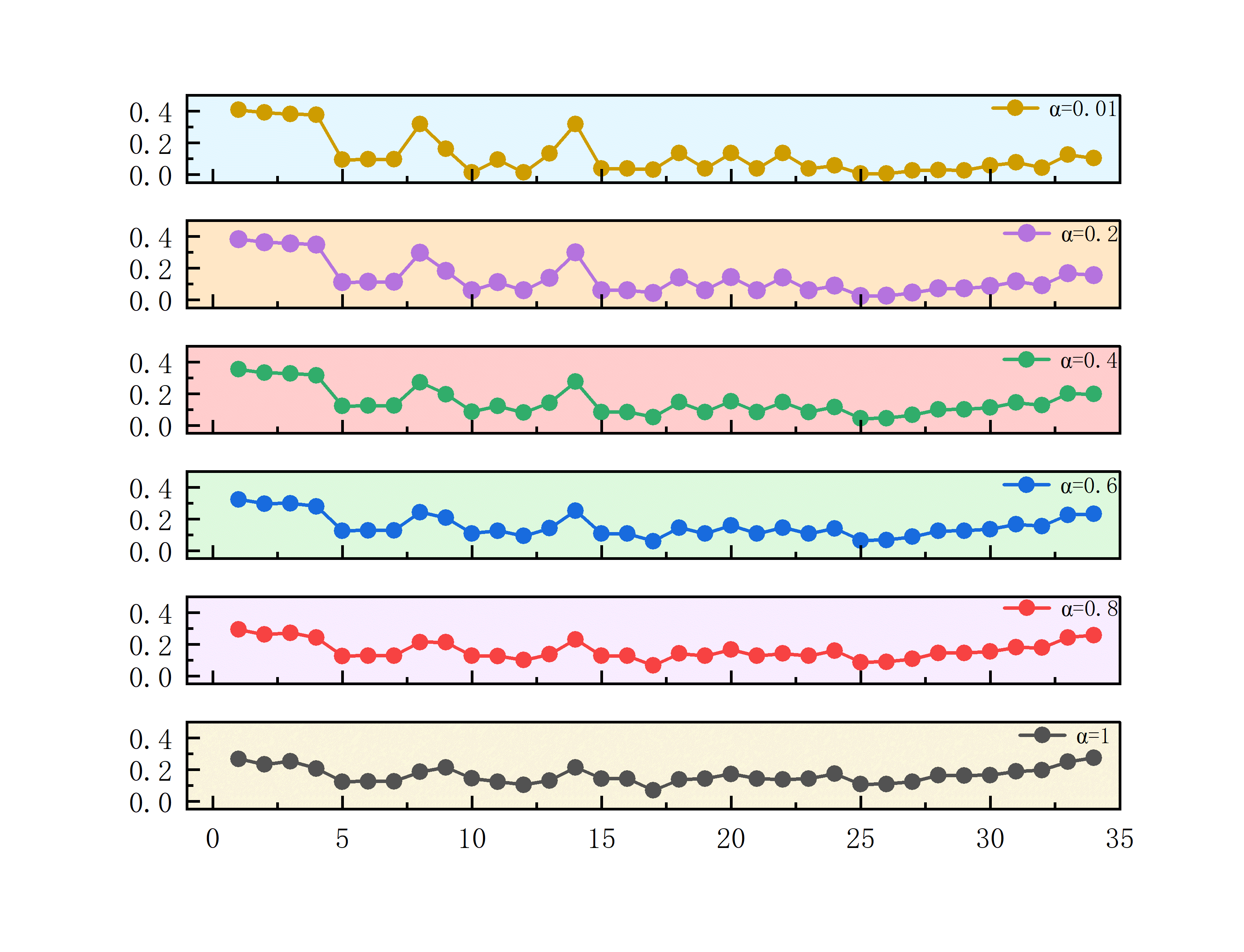}
    \end{minipage}%
  }
  \hfill
  \subfloat[Lusseau's Dolphin network.]{%
    \begin{minipage}[t]{0.47\textwidth}
      \centering
      \includegraphics[width=\linewidth]{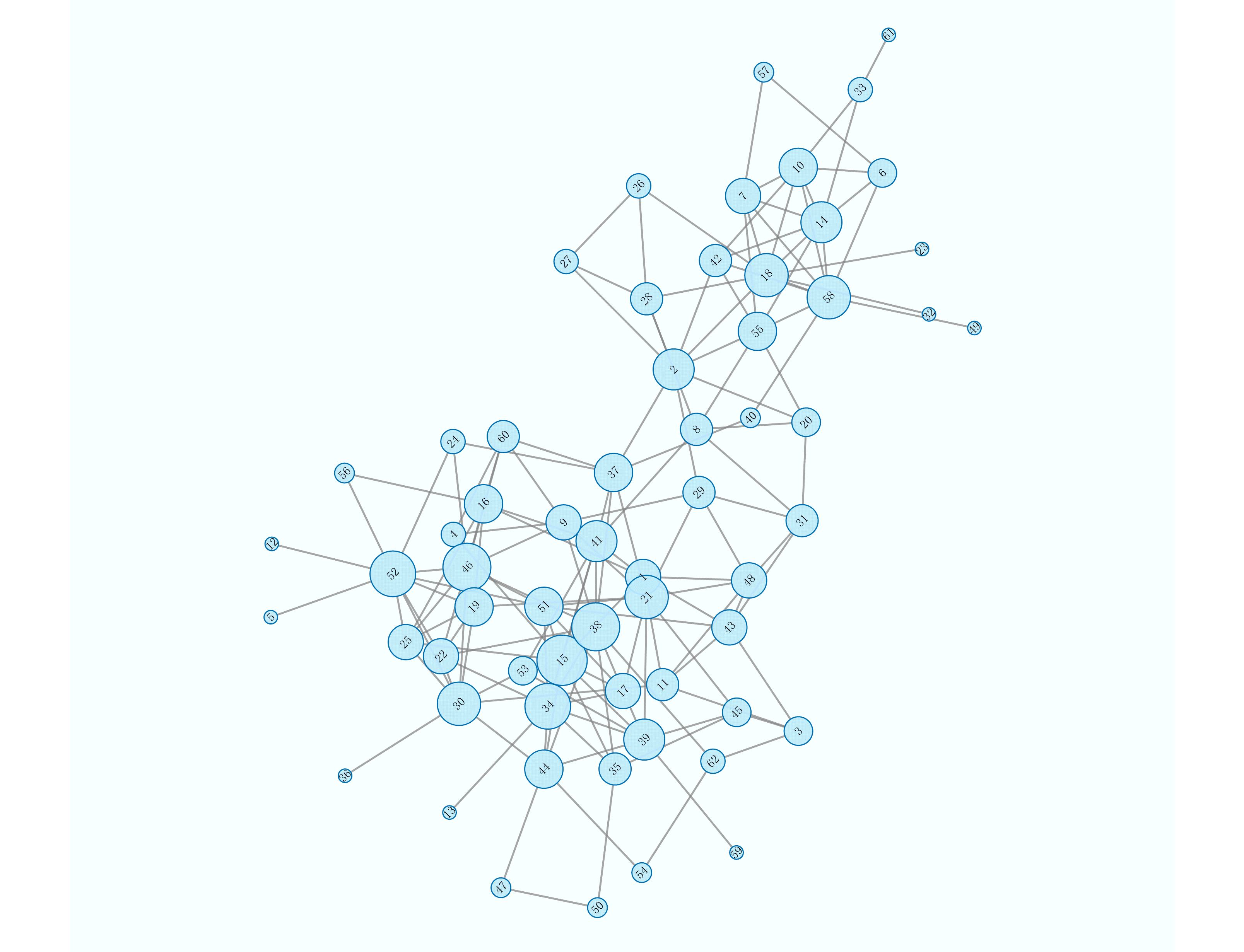}
    \end{minipage}%
  }
    \subfloat[Centrality scores of vertices under different $\alpha$ values.]{%
    \begin{minipage}[t]{0.47\textwidth}
      \centering
      \includegraphics[width=\linewidth]{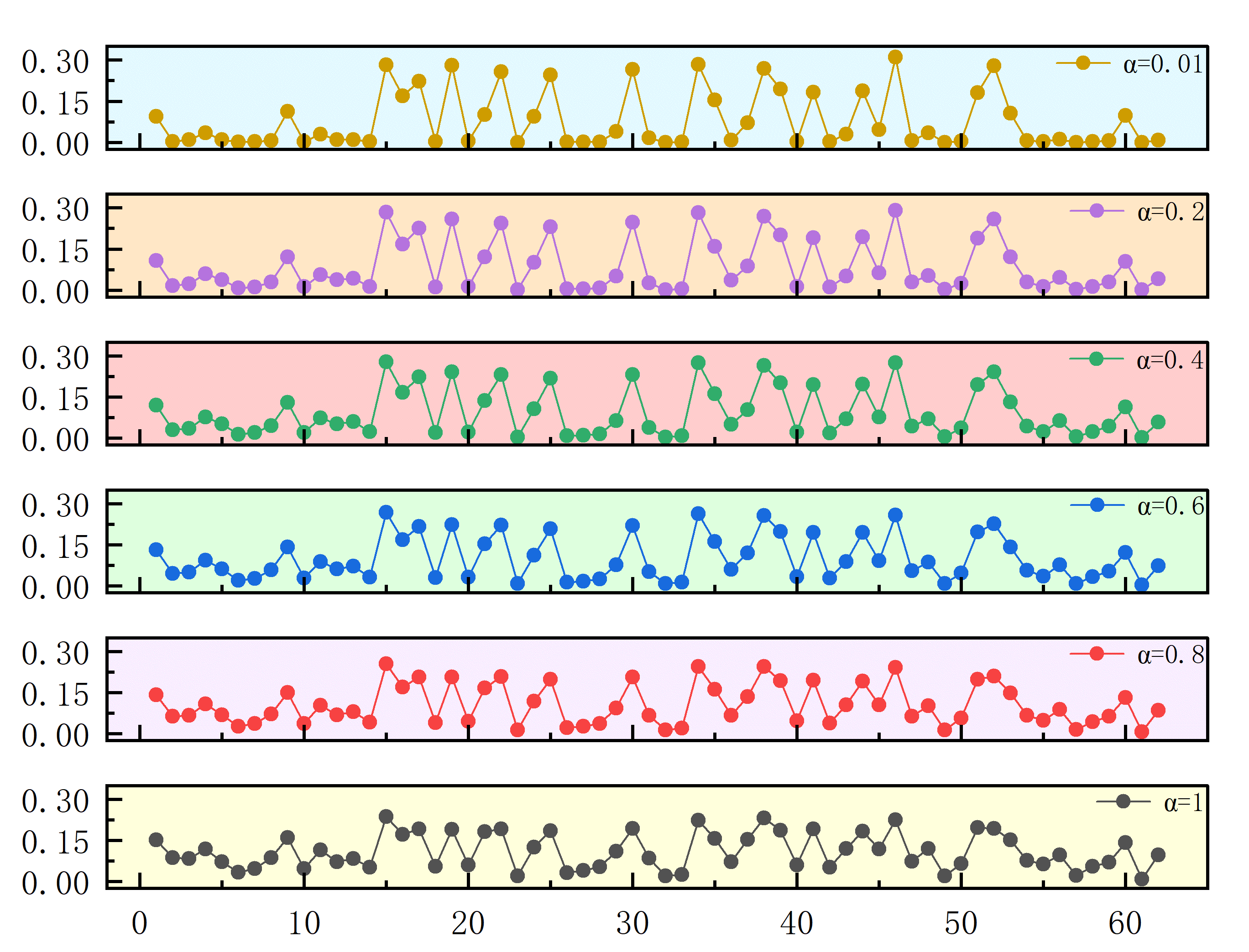}
    \end{minipage}%
  }
  \hfill
  \subfloat[C.elegans metabolic network.]{%
    \begin{minipage}[t]{0.47\textwidth}
      \centering
      \includegraphics[width=\linewidth]{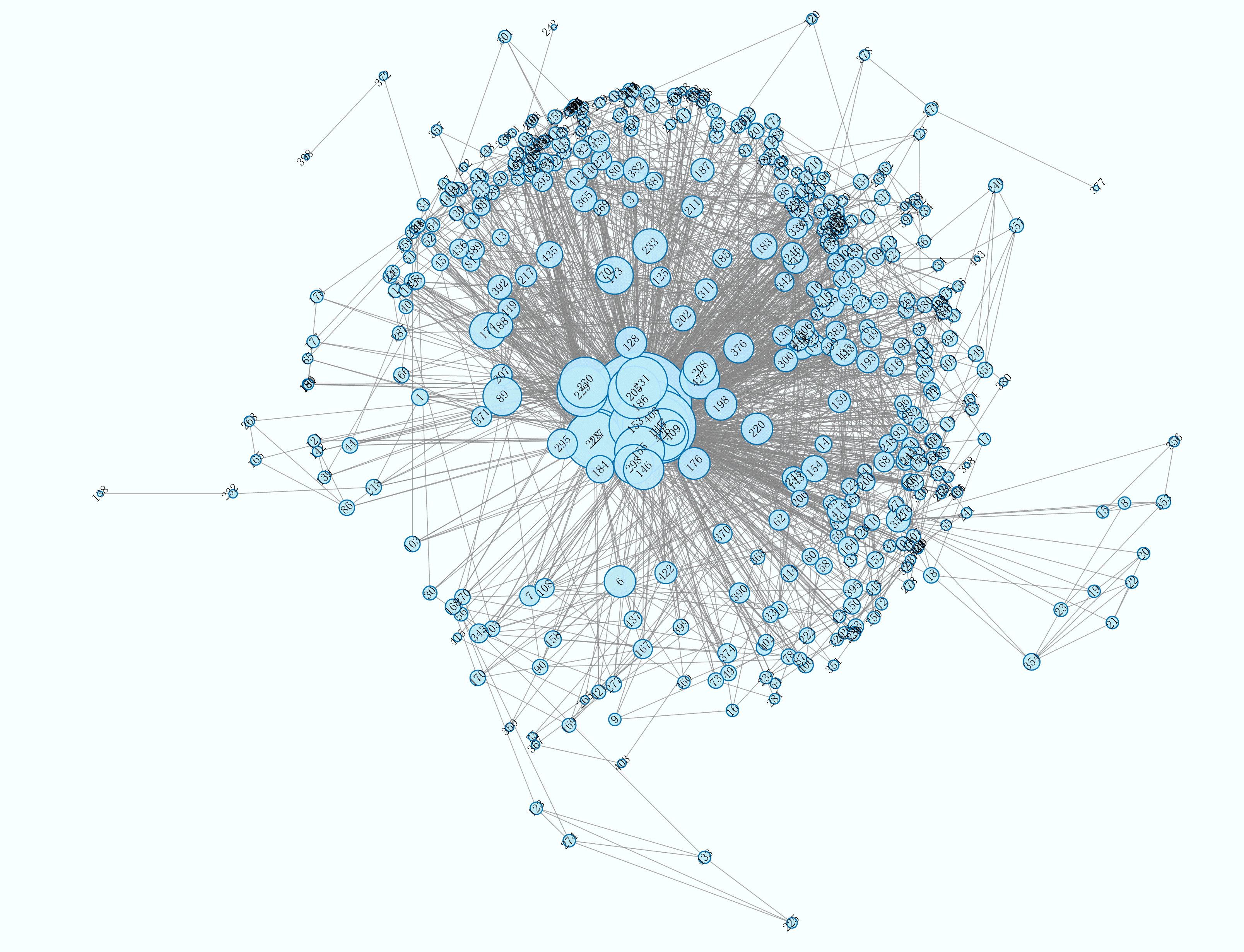}
    \end{minipage}%
  }
    \subfloat[Centrality scores of vertices under different $\alpha$ values.]{%
    \begin{minipage}[t]{0.47\textwidth}
      \centering
      \includegraphics[width=\linewidth]{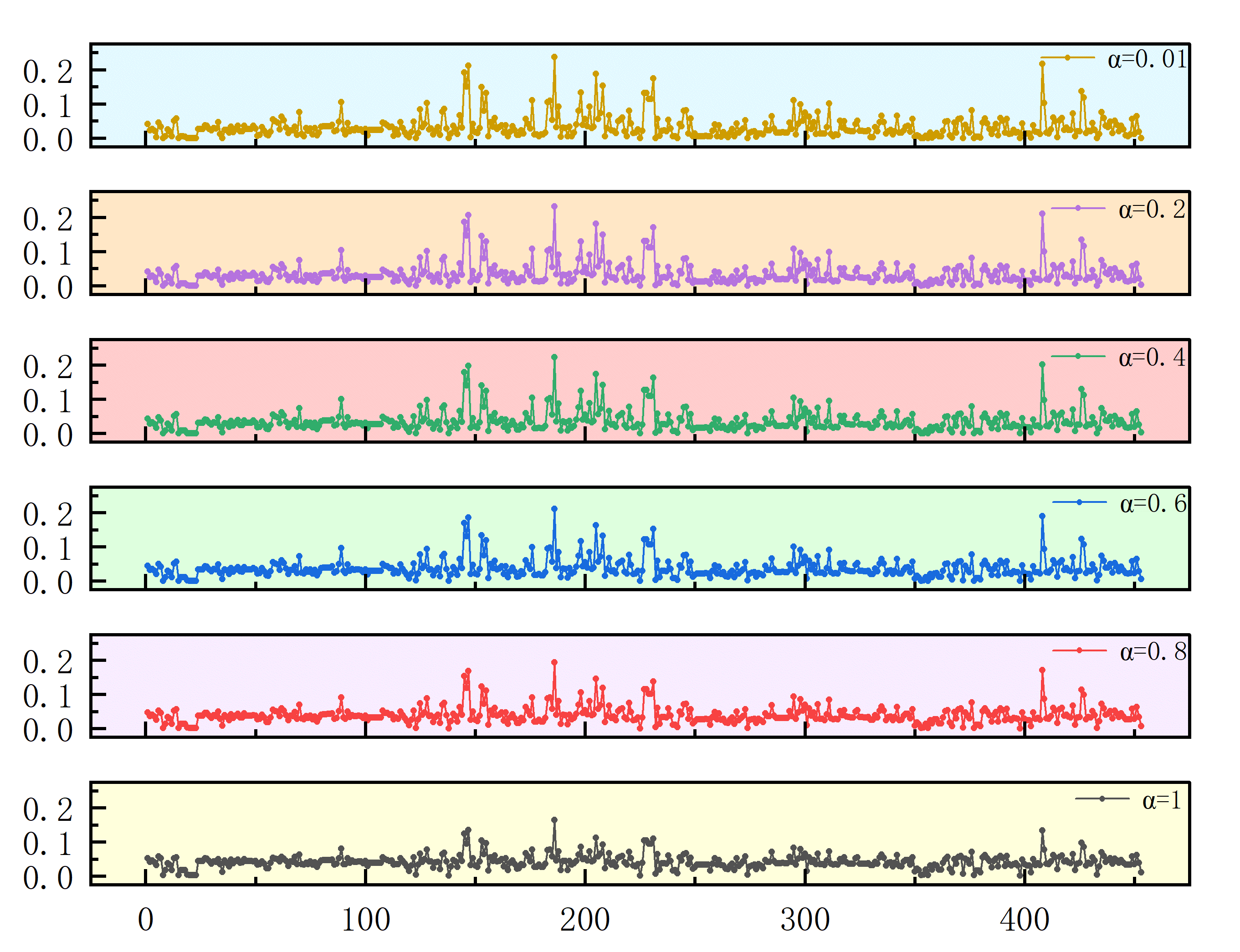}
    \end{minipage}%
  }
  \caption{Three networks and vertices's $\alpha$TEC scores distribution graphs under different values of $\alpha$.}
  \label{fig3}
\end{figure}

\subsubsection{Features of highly ranked vertices}
Below, we conduct a detailed analysis of Zachary's karate club (ZKC) network to explore the characteristics of higher-ranked vertices as the $\alpha$ value decreases from $1$ to $0.01$.
We list the top $10$ vertices in ZKC network under the $\alpha$TEC for different values of $\alpha$ in Table \ref{table1}.
Figure \ref{fig555} illustrates all triangles of ZKC network, and the number of triangles containing vertex
$i$ is denoted as $T(i)$.
From Table \ref{table1}, it can be observed that when $\alpha=1$, vertices $34$ and $1$ rank first and second, respectively. As the value of $\alpha$ decreases, the ranking of vertex $34$ gradually drops out of the top $10$, while vertex $1$ remains consistently in the $1$st place. The rankings of other vertices, such as $9, 20,31$ and $33$, also exhibit significant fluctuations.

From Tables \ref{table1} and Figure \ref{fig555}, we find that vertices with higher rankings in ZKC network typically exhibit larger $T(i)$, except for vertices $20$ and $31$.
Note that, vertices $1$ and $34$ have similar $T(i)$. However, as the value of $\alpha$ decreases, there is a significant difference in their rankings.
Through further comparison, we find that vertices $2, 3, 4, 8, 9$, and $14$ are in triangles whin vertex $1$, all of which also have relatively higher triangle counts and higher centrality scores. In contrast, vertices $9, 33, 15, 16, 19, 21$, and $23$ are in triangles whin vertex $34$. Except for vertices $9$ and $33$, the other vertices have lower triangle counts and correspondingly lower centrality scores.
And we have $T(20)=1$, which is lower than $T(31)=3$. However, it ranks higher than vertex $31$ when $\alpha \leq 0.4$.
Because vertex $20$ is in triangles with vertex $1$ and $2$.
The same situation also occurs with vertices $12$ and $25$, where $T(12)=0$ and $T(25)=1$. However, when $\alpha \leq 0.8$, vertex $12$ has a higher score than vertex $25$. This is because vertex $12$ is influenced by the edge from vertex $1$.
So we have a vertex $i$ with a large $T(i)$ does not necessarily mean that it has a high centrality score. It is also significantly influenced by the vertices that in triangles with $i$.
Moreover, even if a vertex is not in a triangle, if it is adjacent to very important vertices, its score may be higher than those vertices that are in triangles but are not very important.
\begin{table*}[h]
    \tiny
    \centering
    \caption{Under different $\alpha$ values, the top $10$ vertices in  Zachary's karate club network.}
    \label{table1}
    \begin{tabular}{lcccccccccc}  
    \toprule
        \textbf{Different $\alpha$ values} & \textbf{rank1} & \textbf{rank2} & \textbf{rank3} & \textbf{rank4} & \textbf{rank5}
        & \textbf{rank6} & \textbf{rank7} & \textbf{rank8} & \textbf{rank9} & \textbf{rank10}
    \\ \midrule
        $\alpha = 1$        & 34 & 1 & 3 & 33 & 2  & 9 & 14 & 4  & 32  &31 \\
        $\alpha = 0.8$      & 1 & 3 & 2 & 34 & 33 & 4 & 14 & 8  & 9  &31 \\
        $\alpha = 0.6$      & 1 & 3 & 2 & 4 & 14  & 8 & 34 & 33 & 9 & 31 \\
        $\alpha = 0.4$      & 1 & 2 & 3 & 4 & 14 & 8 & 33 & 34 & 9 & 20 \\
        $\alpha = 0.2$      & 1 & 2 & 3 & 4 & 14 & 8 & 9 & 33 & 34 & 20  \\
        $\alpha = 0.01$     & 1 & 2 & 3 & 4 & 14 & 8 & 9 & 20 & 18 & 22  \\
    \bottomrule
    \end{tabular}
\end{table*}

\begin{figure}[H]
\centerline{\includegraphics[scale=0.22]{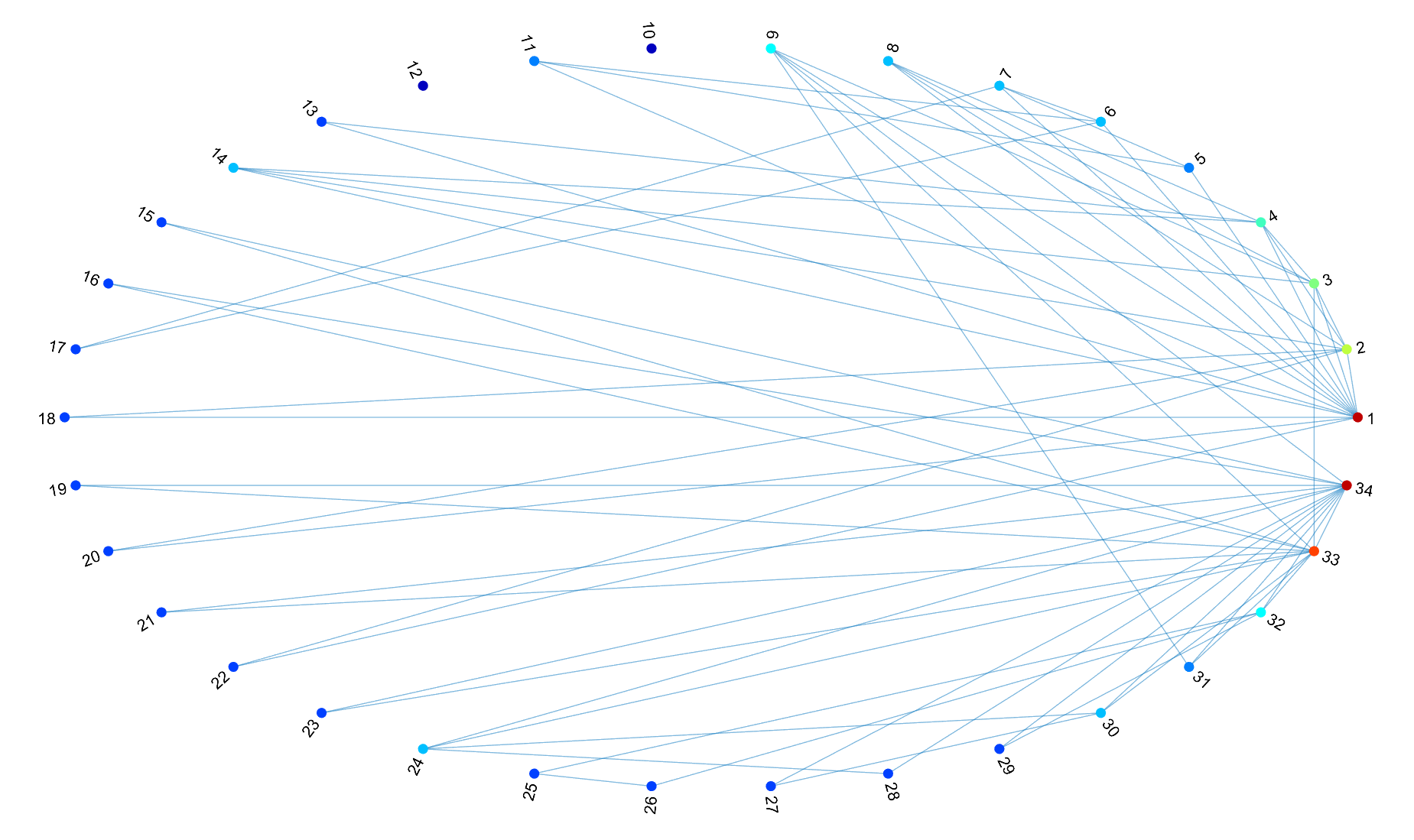}}
\caption{All triangles of Zachary's karate club network.}
\label{fig555}
\end{figure}


\subsubsection{Effect of neighbor triangle count on centrality}
As established in the preceding section, when $\alpha$ assumes smaller values, a vertex's $\alpha$TEC is influenced by vertices sharing triangles with it.
We denote the degree of vertex $i$, the number of triangles containing vertex $i$ and the sum of the number of triangles of the vertices adjacent to $i$ as $D(i)$, $T(i)$, and $NT(i)$, respectively. We plot the distribution graph of $D(i)$, $T(i)$, and $NT(i)$ of the Lusseau's Dolphin network, as shown in Figure \ref{fig555}.

\begin{figure}[H]
\centering
\subfloat[The distribution graph of $D(i)$, $T(i)$, and $NT(i)$.]{
\includegraphics[width=0.46\textwidth]{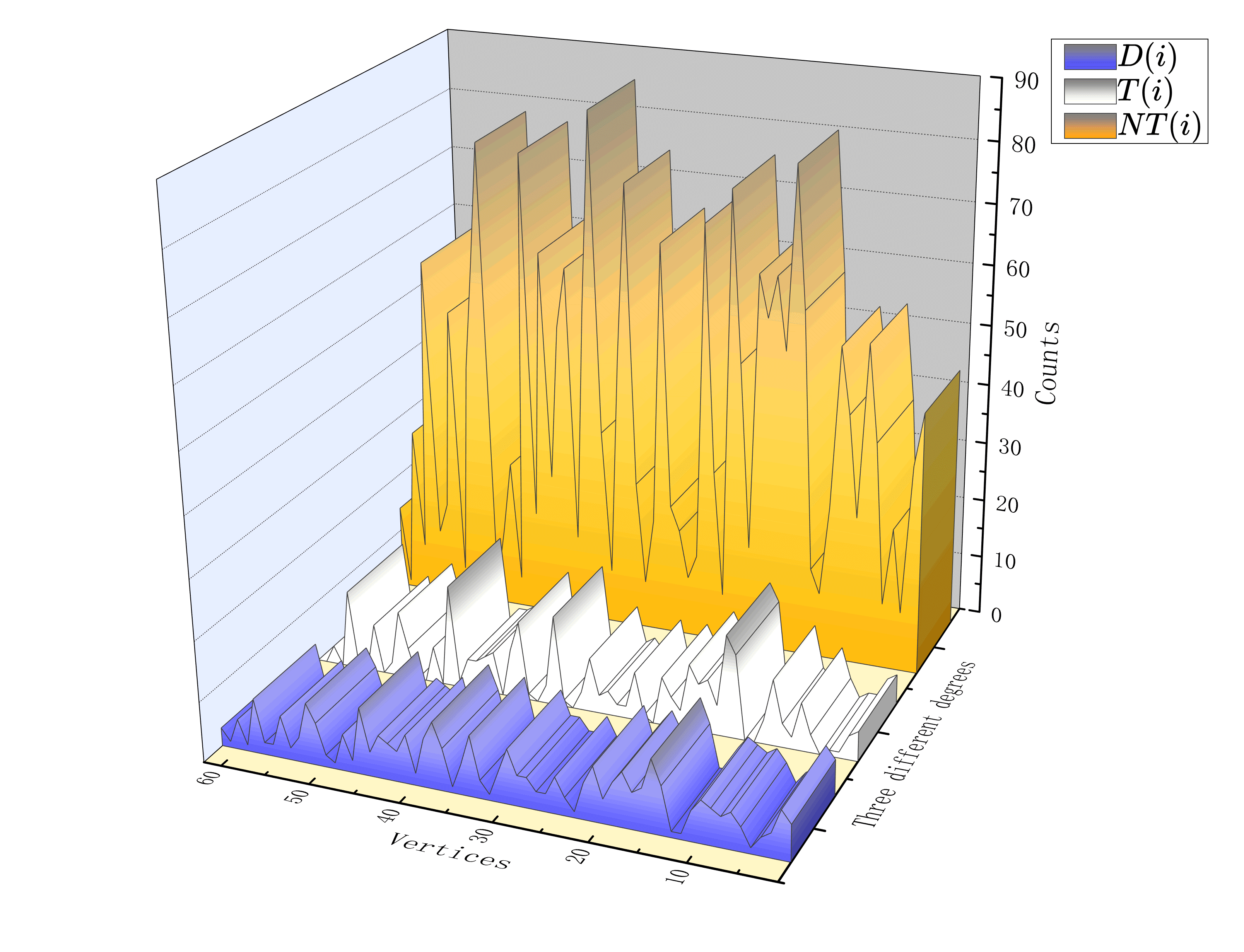}
}
\subfloat[Box plots of $D$, $T$, and $NT$.]{
\includegraphics[width=0.44\textwidth]{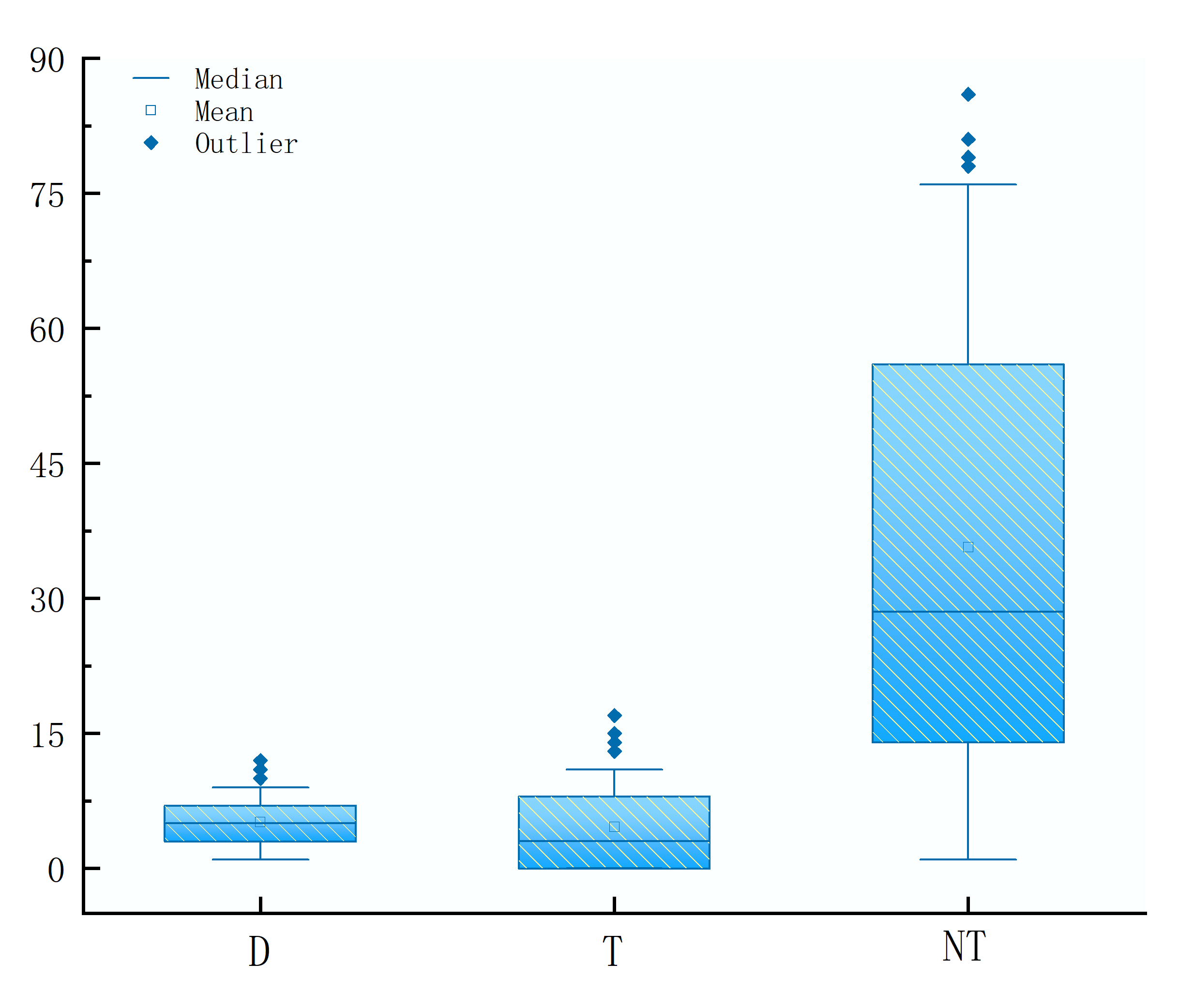}
}
\caption{The distributions of $D(i)$, $T(i)$, and $NT(i)$ in the Lusseau's Dolphin network.}
\label{fig555}
\end{figure}

From Figure \ref{fig555}, we have that the differences in $D(i)$ are not significant, but the differences in $T(i)$ and $NT(i)$ have increased significantly, especially $NT(i)$.
This is also why in Figure \ref{fig3} (d), when $\alpha$ is small, the overall scores are lower, with only a few vertices having higher scores.
Because triangles have a significant impact on centrality, a few vertices have higher $T(i)$ and $NT(i)$ than others.
When $\alpha=1$, the edges have a greater influence on centrality, at this time although some individual vertices have lower scores, the overall difference is not significant.

\subsubsection{Relationship with other centrality measures}
Next, we conduct a correlation analysis between the scores of $\alpha$TEC under four different $\alpha$ values and four types of centrality measures (DC, TC, BC and SC) in the C. elegans metabolic network, as shown in Figure \ref{fig3777}.

\begin{figure}[H]
\centerline{\includegraphics[scale=0.35]{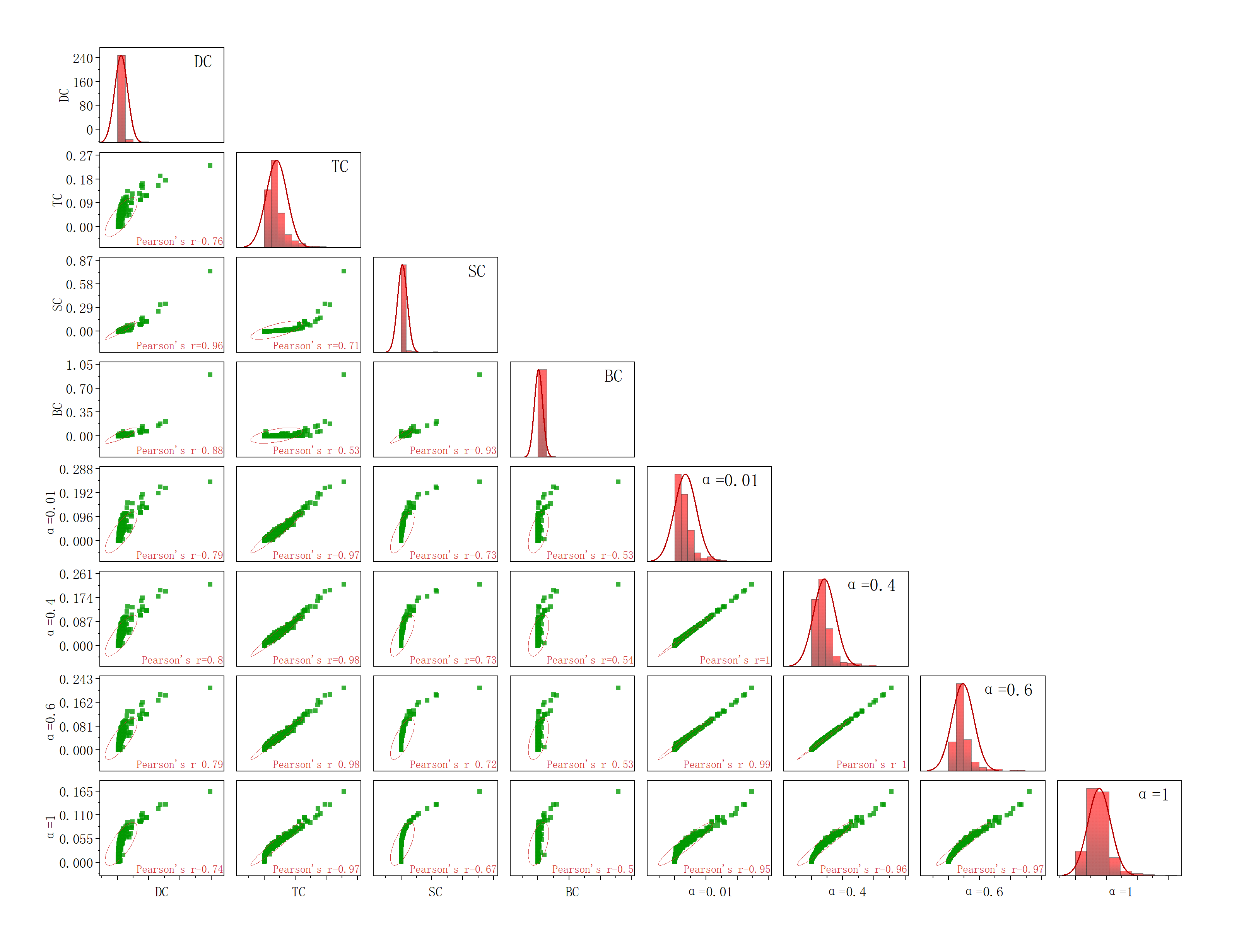}}
\caption{The distribution of vertices's centrality scores in the C. elegans metabolic network under several centrality measures.}
\label{fig3777}
\end{figure}

It can be seen that, except for TC, the other three types of centrality  measures have relatively low correlation with $\alpha$TEC.
Due to the fact that a few vertices in the network have very close links while the majority of vertices have relatively sparse connections, the centrality scores under DC, BC, and SC show significant differences. Under BC and SC, many vertices have a score of $0$.
The centrality measures that do not show such a large gap in the distribution of vertex scores are all centralities that take into account the influence of neighboring points.

When $\alpha=1$, the centrality ranking is consistent with the EC, which takes into account the importance of a vertex's neighbors in determining its centrality. TC considers not only the number of triangles a vertex is in but also the number of triangles its neighbors are in.
Compare to the aforementioned two centralities, the $\alpha$TEC proposed in this paper simultaneously considers the influence of both edges and triangles on centrality, and adjusts the proportion of these influences by tuning $\alpha$. This approach offers greater flexibility and ensures that important information is not lost.

\subsection{Connectivity}
Building upon the proposed $\alpha$-triangle eigenvector centrality ($\alpha$TEC),
we formally define the importance score of each triangle.

\begin{defi}
For a connected graph $G$, let $\boldsymbol {x}$ be the $\alpha$TEC of $G$.
Let $\triangle_i$ be a triangle in $G$.
The importance score of $\triangle_i$ is
\begin{align*}
I_{\triangle_i}= \sum_{p,q,r} x_p+x_q+x_r,
\end{align*}
where $p,q,r$ are vertices of $\triangle_i$, $x_p$ is the $p$th of $\boldsymbol {x}$.
\end{defi}

The critical triangles identified by $\alpha$TEC play a pivotal role in enhancing network Connectivity.
Fiedler vector is the unit eigenvector corresponding to the second smallest Laplacian eigenvalue \cite{fiedler1975property}.
This vector plays a pivotal role in graph-theoretic applications, including community detection in networks \cite{chen2015deep} and analyze the connectivity structure of networks \cite{chen2014local}.
In \cite{jiang2023searching}, the authors proposed a new cycle ranking index $I_{c_i}= \sum_{ (p,q) \in E_i } (x_p - x_q)^2 $ to measure the importance of a cycle in the network,
where $c_i$ is a circle, $E_i$ is the set of all edges of cycle $c_i$ and $x_p$ is the $p$th of the Fiedler vector.
When $c_i$ is a triangle, they give a ranking of triangles under $I_{c_i}$ in the C. elegans  metabolic network.
In the following, we rank triangles under $I_{\triangle_i}$ in the C. elegans metabolic network and compare our results with them.

\begin{table*}[h]
    \tiny
    \centering
    \caption{The ranking of triangles in the C. elegans  metabolic network under $I_{c_i}$ and $I_{\triangle_i}$.}
    \label{table3}
    \begin{tabular}{lccccc}  
    \toprule
        \textbf{Triangle} & \textbf{$I_{c_i}$} & \textbf{Rank} & \textbf{Triangle} & \textbf{$I_{\triangle_i}$} & \textbf{Rank}
    \\ \midrule
       $c_1=[56,153,217]$  & 0.1536  &1   & $\triangle_1=[147,186,408]$   & 0.0418  & 1  \\
       $c_2=[56,123,274]$  & 0.0506  &2   & $\triangle_2=[145,186,408]$   & 0.0405  & 2  \\
       $c_3=[56,274,433]$  & 0.0506  &2   & $\triangle_3=[145,147,186]$   & 0.0402  & 3  \\
       $c_4=[56,123,433]$  & 0.0506  &2   & $\triangle_4=[186,205,408]$   & 0.0402  & 3  \\
       $c_5=[149,154,352]$ & 0.0040  &5   & $\triangle_5=[147,186,205]$   & 0.0399  & 5  \\
        \\ \midrule
        \vdots & \vdots & \vdots & \vdots & \vdots & \vdots \\
         \\ \midrule
       $c_{3280}=[251,254,255]$  & 0  &3284   & $\triangle_{3280}=[8,15,353]$      & 0  & 3284  \\
       $c_{3281}=[51,72,359]$    & 0  &3284   & $\triangle_{3281}=[123,225,274]$   & 0  & 3284  \\
       $c_{3282}=[51,359,447]$   & 0  &3284   & $\triangle_{3282}=[123,225,433]$   & 0  & 3284  \\
       $c_{3283}=[194,219,260]$  & 0  &3284   & $\triangle_{3283}=[225,274,433]$   & 0  & 3284  \\
       $c_{3284}=[252,253,255]$  & 0  &3284   & $\triangle_{3284}=[6,159,169]$     & 0  & 3284  \\
    \bottomrule
    \end{tabular}
\end{table*}

\begin{figure}[H]
\centerline{\includegraphics[scale=0.2]{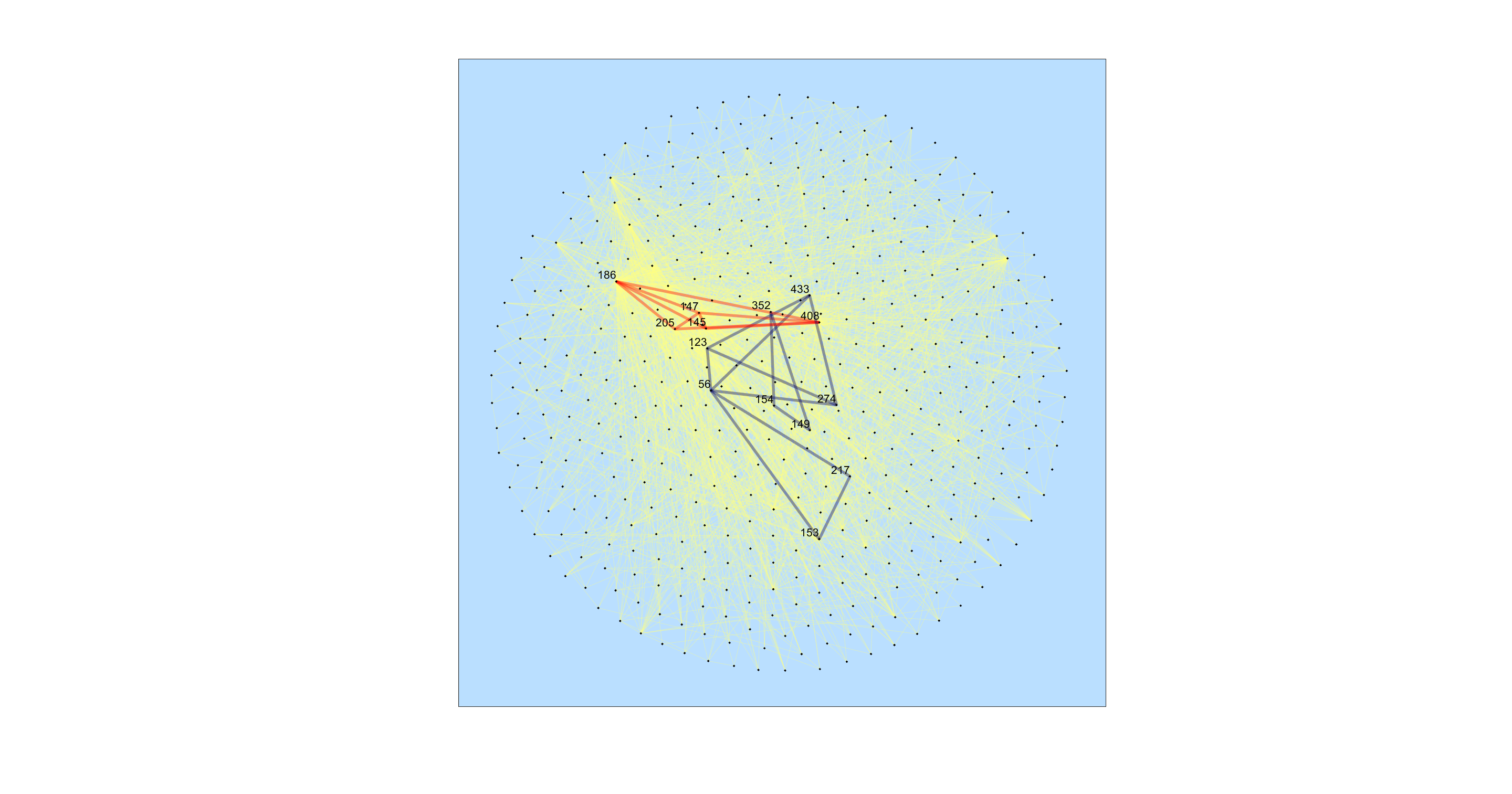}}
\caption{Red(Blue) triangles are important triangles in the C.elegans metabolic network under $I_{\triangle_i}$( $I_{c_i}$).}
\label{fig5}
\end{figure}

Triangles with high centrality index $I_{c_i}$ rankings are typically in bridge positions between communities.
Their removal often increases the number of connected components in networks.
We remove the vertices of $c_1,c_2$ and $c_3$ from the network, respectively. And observe that the originally connected network splits into $3$, $2$, and $2$ components.
In contrast, when we delete the vertices of $\triangle_1,\triangle_2$ and $\triangle_3$, respectively. The network splits into $6$ connected components in all cases.
Notably, vertices $123, 274$, and $433$ exhibit a striking divergence across two measures: they appear in the five most critical triangles under $I_{c_i}$, yet  appear in the five least critical triangles under $I_{\triangle_i}$.
Through comparison, it is evident that the triangles rank as important by the $I_{\triangle_i}$ sorting have a more significant impact on the connectivity of networks.

\section{Conclusion}
The $\alpha$TEC proposed in this paper dynamically adjusts the influence of edges and triangles on centrality by the parameter $\alpha$. Under this centrality measure, each vertex is assigned a nonzero centrality score,
effectively identifies a vertex's structural positioning within the graph, thereby addressing limitations of existing centrality methods.
We prove that for any connected graph, the $\alpha$TEC scores of vertices necessarily exist and are unique.
Furthermore, we provide conditions under which all vertices in a regular graph share identical $\alpha$TEC scores.
Numerical experiments demonstrate that the $\alpha$TEC ranking dynamics of vertices as the $\alpha$ value varies reflect their structural positioning within the graph.
As the value of $\alpha$ varies, the centrality rankings align with the changing of edge structures (stronger influence as $\alpha$ increases) and triangle structures (stronger influence as $\alpha$ decreases).

Additionally, we validate that vertices with higher $\alpha$TEC rankings have significantly influence on network connectivity. In connected graphs, removing triangles ranked highly under $\alpha$TEC splits the graph into $6$ connected components, whereas removing top-ranked triangles identified by another method results in only $3$ or $2$ connected components.

\section*{Acknowledgments}
This work is supported by the National Natural Science Foundation of China (No. 12071097, 12371344), the Natural Science Foundation for The Excellent Youth Scholars of the Heilongjiang Province (No. YQ2022A002) and the Fundamental Research Funds for the Central Universities.

\section*{References}
\bibliographystyle{plain}
\bibliography{ml0ht2}
\end{spacing}
\end{document}